\documentclass[11pt]{utarticle}

\usepackage{amsmath}
\usepackage{mathrsfs}
\usepackage{amsfonts}
\usepackage{amssymb}
\usepackage{amsthm}
\usepackage{mathtools}
\usepackage{graphicx}
\usepackage{color}
\usepackage{ucs}
\usepackage{xparse}
\usepackage{hyperref}
\usepackage{longtable}
\usepackage{xypic}

\long\def\omit#1{}
\definecolor{aqua}{rgb}{0, 1.0, 1.0}
\definecolor{fuschia}{rgb}{1.0, 0, 1.0}
\definecolor{gray}{rgb}{0.502, 0.502, 0.502}
\definecolor{lime}{rgb}{0, 1.0, 0}
\definecolor{maroon}{rgb}{0.502, 0, 0}
\definecolor{navy}{rgb}{0, 0, 0.502}
\definecolor{olive}{rgb}{0.502, 0.502, 0}
\definecolor{purple}{rgb}{0.502, 0, 0.502}
\definecolor{silver}{rgb}{0.753, 0.753, 0.753}
\definecolor{teal}{rgb}{0, 0.502, 0.502}

\makeatletter
\newdimen\itex@wd%
\newdimen\itex@dp%
\newdimen\itex@thd%
\def\itexspace#1#2#3{\itex@wd=#3em%
\itex@wd=0.1\itex@wd%
\itex@dp=#2ex%
\itex@dp=0.1\itex@dp%
\itex@thd=#1ex%
\itex@thd=0.1\itex@thd%
\advance\itex@thd\the\itex@dp%
\makebox[\the\itex@wd]{\rule[-\the\itex@dp]{0cm}{\the\itex@thd}}}
\makeatother

\makeatletter
\newif\if@sup
\newtoks\@sups
\def\append@sup#1{\edef\act{\noexpand\@sups={\the\@sups #1}}\act}%
\def\reset@sup{\@supfalse\@sups={}}%
\def\mk@scripts#1#2{\if #2/ \if@sup ^{\the\@sups}\fi \else%
  \ifx #1_ \if@sup ^{\the\@sups}\reset@sup \fi {}_{#2}%
  \else \append@sup#2 \@suptrue \fi%
  \expandafter\mk@scripts\fi}
\def\tensor#1#2{\reset@sup#1\mk@scripts#2_/}
\def\multiscripts#1#2#3{\reset@sup{}\mk@scripts#1_/#2%
  \reset@sup\mk@scripts#3_/}
\makeatother

\makeatletter
\newbox\slashbox \setbox\slashbox=\hbox{$/$}
\def\itex@pslash#1{\setbox\@tempboxa=\hbox{$#1$}
  \@tempdima=0.5\wd\slashbox \advance\@tempdima 0.5\wd\@tempboxa
  \copy\slashbox \kern-\@tempdima \box\@tempboxa}
\def\slash{\protect\itex@pslash}
\makeatother

\def\clap#1{\hbox to 0pt{\hss#1\hss}}

\def\mathrlap{\mathpalette\mathrlapinternal}

\def\mathrlapinternal#1#2{\rlap{$\mathsurround=0pt#1{#2}$}}

\let\oldroot\root
\def\root#1#2{\oldroot #1 \of{#2}}
\renewcommand{\sqrt}[2][]{\oldroot #1 \of{#2}}

\DeclareSymbolFont{symbolsC}{U}{txsyc}{m}{n}
\SetSymbolFont{symbolsC}{bold}{U}{txsyc}{bx}{n}
\DeclareFontSubstitution{U}{txsyc}{m}{n}

\DeclareSymbolFont{stmry}{U}{stmry}{m}{n}
\SetSymbolFont{stmry}{bold}{U}{stmry}{b}{n}

\DeclareFontFamily{OMX}{MnSymbolE}{}
\DeclareSymbolFont{mnomx}{OMX}{MnSymbolE}{m}{n}
\SetSymbolFont{mnomx}{bold}{OMX}{MnSymbolE}{b}{n}
\DeclareFontShape{OMX}{MnSymbolE}{m}{n}{
    <-6>  MnSymbolE5
   <6-7>  MnSymbolE6
   <7-8>  MnSymbolE7
   <8-9>  MnSymbolE8
   <9-10> MnSymbolE9
  <10-12> MnSymbolE10
  <12->   MnSymbolE12}{}

\makeatletter
\def\re@DeclareMathSymbol#1#2#3#4{%
    \let#1=\undefined
    \DeclareMathSymbol{#1}{#2}{#3}{#4}}
\re@DeclareMathSymbol{\neArrow}{\mathrel}{symbolsC}{116}
\re@DeclareMathSymbol{\neArr}{\mathrel}{symbolsC}{116}
\re@DeclareMathSymbol{\seArrow}{\mathrel}{symbolsC}{117}
\re@DeclareMathSymbol{\seArr}{\mathrel}{symbolsC}{117}
\re@DeclareMathSymbol{\nwArrow}{\mathrel}{symbolsC}{118}
\re@DeclareMathSymbol{\nwArr}{\mathrel}{symbolsC}{118}
\re@DeclareMathSymbol{\swArrow}{\mathrel}{symbolsC}{119}
\re@DeclareMathSymbol{\swArr}{\mathrel}{symbolsC}{119}
\re@DeclareMathSymbol{\nequiv}{\mathrel}{symbolsC}{46}
\re@DeclareMathSymbol{\Perp}{\mathrel}{symbolsC}{121}
\re@DeclareMathSymbol{\Vbar}{\mathrel}{symbolsC}{121}
\re@DeclareMathSymbol{\sslash}{\mathrel}{stmry}{12}
\re@DeclareMathSymbol{\bigsqcap}{\mathop}{stmry}{"64}
\re@DeclareMathSymbol{\biginterleave}{\mathop}{stmry}{"6}
\re@DeclareMathSymbol{\invamp}{\mathrel}{symbolsC}{77}
\re@DeclareMathSymbol{\parr}{\mathrel}{symbolsC}{77}
\makeatother

\makeatletter
\def\Decl@Mn@Delim#1#2#3#4{%
  \if\relax\noexpand#1%
    \let#1\undefined
  \fi
  \DeclareMathDelimiter{#1}{#2}{#3}{#4}{#3}{#4}}
\def\Decl@Mn@Open#1#2#3{\Decl@Mn@Delim{#1}{\mathopen}{#2}{#3}}
\def\Decl@Mn@Close#1#2#3{\Decl@Mn@Delim{#1}{\mathclose}{#2}{#3}}
\Decl@Mn@Open{\llangle}{mnomx}{'164}
\Decl@Mn@Close{\rrangle}{mnomx}{'171}
\Decl@Mn@Open{\lmoustache}{mnomx}{'245}
\Decl@Mn@Close{\rmoustache}{mnomx}{'244}
\makeatother

\makeatletter
\DeclareRobustCommand\widecheck[1]{{\mathpalette\@widecheck{#1}}}
\def\@widecheck#1#2{%
    \setbox\z@\hbox{\m@th$#1#2$}%
    \setbox\tw@\hbox{\m@th$#1%
       \widehat{%
          \vrule\@width\z@\@height\ht\z@
          \vrule\@height\z@\@width\wd\z@}$}%
    \dp\tw@-\ht\z@
    \@tempdima\ht\z@ \advance\@tempdima2\ht\tw@ \divide\@tempdima\thr@@
    \setbox\tw@\hbox{%
       \raise\@tempdima\hbox{\scalebox{1}[-1]{\lower\@tempdima\box
\tw@}}}%
    {\ooalign{\box\tw@ \cr \box\z@}}}
\makeatother

\makeatletter
\NewDocumentCommand\mathraisebox{moom}{%
\IfNoValueTF{#2}{\def\@temp##1##2{\raisebox{#1}{$\m@th##1##2$}}}{%
\IfNoValueTF{#3}{\def\@temp##1##2{\raisebox{#1}[#2]{$\m@th##1##2$}}%
}{\def\@temp##1##2{\raisebox{#1}[#2][#3]{$\m@th##1##2$}}}}%
\mathpalette\@temp{#4}}
\makeatletter

\makeatletter
\def\udots{\mathinner{\mkern2mu\raise\p@\hbox{.}
\mkern2mu\raise4\p@\hbox{.}\mkern1mu
\raise7\p@\vbox{\kern7\p@\hbox{.}}\mkern1mu}}
\makeatother

\newcommand{\lt}{<}

\newcommand{\integral}{\int}

\theoremstyle{plain}

\newtheorem*{utheorem}{Theorem}
\newtheorem*{ulemma}{Lemma}

\theoremstyle{definition}

\theoremstyle{remark}

\newcommand{\Tr}[2][]{\mathrm{Tr}_{#1}\left(#2\right)}
\def\Bid{{\mathchoice {\rm {1\mskip-4.5mu l}} {\rm
{1\mskip-4.5mu l}} {\rm {1\mskip-3.8mu l}} {\rm {1\mskip-4.3mu l}}}}

\begin{document}
\preprint{
UTTG-01-17\\}
\title{von Neumann's Formula, Measurements and the Lindblad Equation}

\author{Jacques Distler and Sonia Paban
     \oneaddress{
      Theory Group \\
      Department of Physics,\\
      University of Texas at Austin,\\
      Austin, TX 78712, USA \\
      {~}\\
      \email{distler@golem.ph.utexas.edu}\\
      \email{paban@physics.utexas.edu}\\
      }
}

\date{\today}

\Abstract{
We previously remarked that when an observable $A$ has a continuous spectrum, then von Neumann's formula  for the post-measurement state needs to be extended and the correct formula ineluctably involves the resolution of the detector used in the measurement. We generalize previous results to compute the uncertainties in successive measurements of more general pairs of observables. We also show that this extended von Neumann's formula for the  post-measurement state is a completely positive map and, moreover, that there is a completely-positive interpolation between the pre- and post-measurement states.  Weinberg has advocated that the time-evolution \emph{during} the measurement process should be modeled as an open quantum system and governed by a Lindblad equation. We verify that this is indeed the case for an arbitrary observable, $A$, and a fairly general class of interpolations.
}

\maketitle

\thispagestyle{empty}
\vfill
\newpage
\thispagestyle{empty}
\tableofcontents
\vfill
\newpage
\setcounter{page}{1}

\section{Introduction}

Our previous paper \cite{Distler:2012eh} was devoted to the topic of uncertainties in successive measurements: first measure an observable, $A$, and then (in the post-measurement state) measure a non-commuting observable, $B$. One can prove lower bounds for the product, $\Delta A\Delta B$, of the uncertainties. When $A$ has a continuous spectrum (say, $A=x$), then von Neumann's formula \eqref{vonN1} for the post-measurement state needs to be extended and the correct formula ineluctably involves the resolution of the detector used in the measurement. In \cite{Distler:2012eh}, we developed the formalism to describe this only for the special case of $A=x$ and $B=p$. Here, we present the general formula \eqref{vonN2} for an \emph{arbitrary} observable.

As an application this allows us, in \S\ref{genHeisenberg}, to extend our previous results to successive measurements of more general pairs of observables.  But, along the way, we make progress in some other directions.

In \S\ref{proveCompletePositivity} we show that   \eqref{vonN2} is a completely positive map and, moreover, that there is a completely-positive interpolation between the pre- and post-measurement states.

In \cite{Weinberg:2016axv}, Weinberg advocated that the time-evolution \emph{during} the measurement process should be modeled as an open quantum system and governed by a Lindblad equation \cite{Lindblad:1975ef, Gorini:1975nb} (see also \cite{Bassi:2003gd} and references therein for other approaches which lead to similar dynamics). In \S\ref{showLindblad} we verify that this is indeed the case for an arbitrary observable, $A$, and a fairly general class of interpolations.

\section{Measuring a General Observable}

Let $A$ be a self-adjoint operator with a pure point spectrum and let $P_\lambda$ be the projection onto the $\lambda$-eigenspace of $A$. The Born rule gives the probability that,  in state $\rho$, $A$ has value $\lambda$:
\begin{equation}\label{Born1}
\text{Prob}_\rho(\lambda) = \Tr{P_\lambda \rho}
\end{equation}
\emph{Measuring} $A$ changes the state, $\rho$, in a non-unitary way, given by von Neumann's formula \cite{vonNeumann},
\begin{equation}\label{vonN1}
\rho'  = \sum_\lambda P_\lambda \rho P_\lambda
\end{equation}
How do these formul\ae\ change when $A$ doesn't have a pure point spectrum? The modification to the Born rule is well-known, but the modification to von Neumann's formula is not. The latter ineluctably involves the additional data of a finite detector resolution and is most elegantly stated in the language (see, e.g., \cite{peres} for an introduction) of positive operator-valued measures (POVMs).

Before stating both of these, a brief digression on the Spectral Theorem (which will serve to fix our notation) is in order.

\section{The Spectral Theorem}

The Spectral Theorem states that there is a 1-1 correspondence between self-adjoint operators and projection-valued measures (PVMs),
\begin{equation}\label{spectral}
   A \leftrightarrow dP_\lambda
\end{equation}
A projection-valued measure (see, e.g., \cite{Mackey}), $dP_\lambda$, assigns to every Borel subset $S\subset \mathbb{R}$ a projection operator
\begin{equation}
\Pi(S) = \int_S dP_\lambda
\end{equation}
where the collection of projection operators, $\Pi(S)$, satisfy

\begin{itemize}
\item $\Pi(\emptyset) = 0$
\item $\Pi( \mathbb{R}) = 1$
\item $\Pi( S_1 \cap S_2) = \Pi( S_1)\Pi(S_2)$
\item If $S_1\cap S_2 =\emptyset$,  then $\Pi( S_1 \cup S_2) = \Pi( S_1)+\Pi(S_2)$\quad .
\end{itemize}

\noindent
The correspondence \eqref{spectral} is given by
$$
   A = \int \lambda dP_\lambda
$$

\section{Born Rule and von Neumann's Formula: the General Case}
The generalization of the Born rule \eqref{Born1} states that, in the state $\rho$, the probability that $A$ has value in $S$ is given by
\begin{equation}\label{Born2}
\text{Prob}_\rho(S) = \Tr{\Pi(S)\rho}
\end{equation}
When $A$ has a pure point spectrum, $dP_\lambda$ is sum of delta-functions and \eqref{Born2} reduces to \eqref{Born1}.

What about von Neumann's formula? In order to come up with a suitable generalization of \eqref{vonN1}, we \emph{need} to introduce a finite detector resolution and to draw a distinction between the intrinsic quantum-mechanical distribution \eqref{Born2}, for the values of $A$, and the probability distribution for the \emph{measured} values for $A$.

Let $|f(l)|^2 d l$ be a probability distribution on the real line with mean zero and standard deviation $\sigma_A$,
\begin{equation}\label{fdist}
\begin{split}
1 &= \int |f(l)|^2 d l\\
0 &= \int l |f(l)|^2 d l\\
\sigma_A^2 &= \int l^2 |f(l)|^2 d l
\end{split}
\end{equation}
We will interpret $|f(l)|^2 d l$ as an ``acceptance function" representing the finite resolution of our detector. You might keep in mind a Gaussian,
$$
f(l) = \frac{1}{(2\pi\sigma_A^2)^{\mathrlap{1/4}}}\quad e^{-l^2/4\sigma_A^2}
$$
but any smooth function satisfying \eqref{fdist} will do. Indeed, it is sometimes convenient to assume that $f(l)$ is also compactly-supported. An example to keep in mind is
\[
\begin{split}
f(l)&=
\begin{cases}
\frac{1}{\sqrt{c_1 \sigma_A}} \exp\left({-\frac{c_0 \sigma_A^2}{4 \sigma_A^2-l^2}}\right)
&|l|<2\sigma_A\\
 0& |l|\geq 2\sigma_A
\end{cases}\\
c_0 &\approx 0.389538720\\
c_1&\approx 2.711138019
\end{split}
\]

The probability that measuring $A$ yields a value in the Borel subset $S$ is
\begin{equation}\label{Born3}
\text{Prob}_{\text{meas}}(S) = \Tr{F(S)\rho} 
\end{equation}
where
\begin{equation}
  F(S) = \int_S F_l dl
\end{equation}
and the POVM $F_l dl$ is given by
\begin{equation}
   F_l = \int |f(l-\lambda)|^2 dP_\lambda = |f(l-A)|^2
\end{equation}
The positive (more precisely, positive \emph{semi}-definite) operators, $F(S)$ satisfy
\begin{itemize}
\item $F(\emptyset) = 0$
\item $F( \mathbb{R}) = 1$
\item If $S_1\cap S_2 =\emptyset$,  then $F( S_1 \cup S_2) = F( S_1)+F(S_2)$
\end{itemize}

If $A$ has a pure point spectrum and the support of \eqref{fdist} is smaller than the inter-eigenvalue spacing, then \eqref{Born3} reduces to \eqref{Born1} and there is no distinction between the measured and ``intrinsic" distributions of values for $A$. But, even in the case of a pure point spectrum, we might be saddled with a detector too crude to resolve some closely-spaced eigenvalues and would want to distinguish \eqref{Born1} from \eqref{Born3}.

The smearing of the intrinsic probability distribution \eqref{Born2} against the finite detector resolution  \eqref{fdist}, to yield \eqref{Born3}, has the obvious desired property that the intrinsic quantum-mechanical uncertainty in $A$ adds in quadrature with the detector resolution to give the \emph{measured} uncertainty
\begin{equation}
\begin{split}
(\Delta A)_{\text{meas}}^2 &= \int l^2 \Tr{F_l \rho}dl  - \left(\int  l \Tr{F_l \rho}dl\right)^2\\
&= \Tr{A^2\rho} -(\Tr{A\rho})^2 + \sigma_A^2\\
&= (\Delta A)_\rho^2 + \sigma_A^2
\end{split}
\end{equation}

Now we can state the generalization of \eqref{vonN1}. Let
\begin{equation}\label{Ldef}
F_l = L^\dagger_l L_l
\end{equation}
This splitting is additional data. We will take
\begin{equation}\label{Ldef2}
 L_l = \int f(l-\lambda) dP_\lambda = f(l-A)
\end{equation}
but you can redefine $L_l\to U L_l$ for any unitary $U$ without changing \eqref{Ldef}

von Neumann's formula now states that the post-measurement state is 
\begin{equation}\label{vonN2}
\rho' = \int dl\,  L_l \rho L_l^\dagger = \int dl\,  f(l-A) \rho f(l-A)^\dagger
\end{equation}
Note that

\begin{itemize}
\item As with \eqref{vonN1}, the probability distribution for $A$ is the \emph{same} in the post-measurement state, $\rho'$, as it was in the state $\rho$. However, the probability distributions for other (non-commuting) observables are affected by \eqref{vonN2}.
\item The size of the effect depends on the detector resolution \eqref{fdist}.
\item Whereas with \eqref{vonN1}, repeated measurements of $A$ do not (further) affect the state, in the case of a continuous spectrum repeated measurements of $A$, via  \eqref{vonN2}, continue to alter the post-measurement state.
\item As with \eqref{vonN1}, we can extend \eqref{vonN2} to handle questions about conditional probabilities: ``\emph{Given} that a measurement of $A$ yielded a value $a\in S\subset \mathbb{R}$, what is the probability that \dots?" To compute conditional probabilities, we use the Born rule in conjunction with the conditional density matrix,
\begin{equation}
\rho_S = \frac{\int_S dl\, f(l-A) \rho f(l-A)^\dagger}{\Tr{F(S)\rho}}
\end{equation}
\item The map from $\rho\to\rho'$ is a trace-preserving completely-positive map.
\end{itemize}

This last point is important, and bears remarking upon.

\section{Complete Positivity}\label{proveCompletePositivity}

A map $\Phi: B(\mathcal{H})\circlearrowleft$ is \emph{completely-positive} if, for every finite-dimensional $V$, 
$$\Phi\otimes \Bid_V: B(\mathcal{H}\otimes V)\circlearrowleft$$
is a positive map.

Stinespring's Theorem provides a characterization of when a positive map is completely positive.

\begin{utheorem}[Stinespring]
$\Phi: B(\mathcal{H})\circlearrowleft$ is completely positive iff there exists a Hilbert space $\mathcal{K}$, a unital $*$-homomorphism, $s:  B(\mathcal{H})\to B(\mathcal{K})$, and a bounded operator $W: \mathcal{H}\to  \mathcal{K}$ such that $\Phi$ can be written in the form
$$
 \Phi(\rho) = W^\dagger s(\rho) W
$$
\end{utheorem}

In the case of \eqref{vonN2}, the conditions of Stinespring's Theorem are satisfied with

\begin{itemize}
\item $\mathcal{K}= \mathcal{H}\otimes L^2(\mathbb{R})$
\item $s(\rho)= \rho\otimes \Bid$
\item $W: |\psi\rangle\mapsto f(l-A)^\dagger|\psi\rangle$
\item $W^\dagger: |\psi(l)\rangle \mapsto \int dl f(l-A) |\psi(l)\rangle$
\end{itemize}

We can do even better. There exists a completely-positive, trace-preserving \emph{evolution} from $\rho$ to $\rho'$. Let $A(\epsilon)$ be any 1-parameter family of self-adjoint operators which interpolate between $A(0) = \overline{\lambda} \Bid$ (where $\overline{\lambda}$ is a constant) and $A(1) = A$. We could, for instance, choose
$$
A(\epsilon) = (1-\epsilon) \overline{\lambda} \Bid + \epsilon A
$$
but any interpolation will do. Then
\begin{equation}\label{evolr}
\rho(\epsilon) = \integral dl f(l-A(\epsilon)) \rho f(l-A(\epsilon))^\dagger
\end{equation}
is a 1-parameter family of density matrices which interpolates between $\rho(0)=\rho$ and $\rho(1)=\rho'$. Moreover, by Stinespring's Theorem, the map $\Phi(\epsilon): \rho\mapsto\rho(\epsilon)$ is completely-positive. 

We can use this to describe how the density matrix evolves \emph{during} the measurement. Let us use our freedom to redefine $L_l\to U L_l$ to write
\begin{equation}\label{measA}
\rho(t) = \int e^{-i Ht} f(l-\widetilde{A}(t)) \rho(0) f(l-\widetilde{A}(t))^\dagger e^{i Ht}\, dl
\end{equation}
where $\widetilde{A}(t)$  is a time-dependent self-adjoint operator which interpolates between $\widetilde{A}(0)= \overline{\lambda} \Bid$ and $\widetilde{A}(t_f)=A$:
\begin{equation}\label{tildeAdef}
\begin{split}
   \widetilde{A}(t) &= \begin{cases}\overline{\lambda} \Bid &t<0\\A&t>t_f\end{cases}\\
  \end{split}
\end{equation}
For $t<0$ and $t>t_f$, $\rho(t)$ evolves unitarily with Hamiltonian $H$. During the time-interval $t\in[0,t_f]$, it evolves from the pre- to the post-measurement state as dictated by (our generalization of) von Neumann's formula. Moreover, the map $\Phi(t,0): \rho(0)\to\rho(t)$ is completely-positive for all $t$.

In \cite{Weinberg:2016axv}, Weinberg demanded a \emph{stronger} condition, namely that the map $\Phi(t+\delta t,t): \rho(t)\to\rho(t+\delta t)$ be completely-positive for all $t$. Equivalently, he demanded that $\rho(t)$ satisfy a Lindblad equation. As we shall see in the next section, this imposes constraints on the form of the interpolation $\widetilde{A}(t)$. While we don't have a general characterization of the interpolations which satisfy a Lindblad equation, we will find a broad class which \emph{do}.

\section{The Lindblad equation}\label{showLindblad}

For our detector acceptance function we will take a Gaussian,
\begin{equation}\label{Egaussian}
f(l) = \frac{1}{(2\pi\sigma_A^2)^{\mathrlap{1/4}}}\quad e^{-l^2/4\sigma_A^2}
\end{equation}
The interpolation $\widetilde{A}(t)$ can be fairly arbitrary, but we will require that $[\widetilde{A}(t), A]=0$, which is tantamount to demanding that the measuring apparatus not disturb the probability distribution for $A$ (the distribution might still evolve because $A$ may not commute with $H$, but that's the only source of its evolution). For generic $A$, this implies $[\widetilde{A}(t),\dot{\tilde{A}}(t)]=0$.

It is not hard to show (see Appendix \ref{proveLindblad}) that under these conditions \eqref{measA} is the solution to the  equation

\begin{equation}\label{generaleq}
\dot{\rho} = -i[H,\rho]+\frac{1}{4 \sigma_A^2} \left(X \rho Y + Y \rho X- \{XY,\rho\} \right)
\end{equation}
where
\begin{equation}\label{XYdef}
\begin{split}
  X &= e^{-i Ht} \widetilde{A} e^{i H t}\\
  Y &= e^{-i Ht} \dot{\tilde{A} }e^{i H t}\\ 
\end{split}
\end{equation}
For a generic interpolation, $\widetilde{A}(t)$, this is \emph{not} of Lindblad form. We can rewrite it as
\begin{equation}\label{rewrite}
\dot{\rho} = -i[H,\rho]+D[\tfrac{1}{4\sigma_A}(X+Y)]\rho - D[\tfrac{1}{4\sigma_A}(X-Y)]\rho
\end{equation}
where $D[L]:\rho \to L\rho L^\dagger - \tfrac{1}{2}\{L^\dagger L,\rho\}$ is the Lindblad superoperator. Lindblad's equation requires a sum over $D[L_i]\rho$ with \emph{positive} coefficients (which, by rescaling the $L_i$, we can take to be 1).

However, if we choose a linear interpolation,
\begin{equation}\label{linAt}
\widetilde{A}(t) = (1-\epsilon(t))\overline{\lambda} + \epsilon(t) A,
\end{equation}
where $\epsilon(t)$ interpolates between $0$ and $1$ as 
\begin{equation}\label{epsilondef}
\begin{split}
   \epsilon(t) &= \begin{cases}0&t<0\\1&t>t_f\end{cases}\\
   \dot{\epsilon}(t) &\geq 0,\quad \forall t
\end{split}
\end{equation}
the equation \eqref{generaleq} reduces to a time-dependent Lindblad equation,
\begin{equation}
\dot{\rho} = -i[H,\rho]+ D[L(t)]\rho
\end{equation}
with
\begin{equation}
L(t) = \tfrac{(2\epsilon\dot{\epsilon})^{\mathrlap{1/2}}}{2\sigma_A}\quad e^{-i Ht} A e^{i H t}
\end{equation}

Thus we see that any linear interpolation \eqref{linAt} will satisfy a Lindblad equation for a monotonic, but otherwise \emph{arbitrary} function $\epsilon(t)$. If we relax the monotonicity assumption in \eqref{epsilondef}, this will fail to be true, even though $\rho(0)\to\rho(t)$ is still a completely-positive map.

An example might be helpful to illustrate the distinction. Let $\mathcal{H}=L^2(\mathbb{R})$ and let us measure the observable $x$. For present purposes, let us neglect the Hamiltonian, $H$, in writing the evolution of $\rho(t)$ during the measurement. We can represent $\rho(t)$ as an integral kernel,
$$
\rho(t): g(x)\mapsto \int \rho(x,y;t) g(y) dy
$$
In \cite{Distler:2012eh},  we showed that
$$
\rho(x,y;t) = e^{-\epsilon(t)^2 (x-y)^2/8\sigma_x^2}  \rho(x,y;0)
$$
The map $\rho(0)\to\rho(t)$ is completely positive. But now, it is easy to write the map for arbitrary initial and final times,
$$
\rho(x,y;t_2) = e^{-\left(\epsilon(t_2)^2-\epsilon(t_1)^2\right) (x-y)^2/8\sigma_x^2}  \rho(x,y;t_1)
$$
The map $\rho(t_1)\to\rho(t_2)$ is only positive (let alone completely-positive) if
$$
\epsilon(t_2)^2 \geq \epsilon(t_1)^2
$$
Since $\epsilon(0)=0$, this is trivially-satisfied if we restrict ourselves to imposing this condition only for $t_1=0$. Demanding complete positivity of $\rho(t)\to\rho(t+\delta t)$ at every instant in time (or, equivalently, for all $t_2>t_1$) enforces the stronger condition, namely that $\epsilon(t)$ be monotonic in time.

We will leave to future work the question of what conditions need to be imposed for a detector acceptance function, $f(l)$, which is not Gaussian, or an interpolation $\widetilde{A}(t)$ which is not linear.

\section{Application: Heisenberg's Uncertainty Relation and its Generalizations}\label{genHeisenberg}

In Heisenberg's original paper \cite{Heisenberg}, he derived the relation
$$
   \Delta x \Delta p \geq 1
$$
from a gedanken experiment where one measures the position and momentum of an electron. Localizing the electron, very well, requires a high-energy probe, which imparts a large kick -- and hence a large uncertainty in the momentum -- to the electron. We will see that his relation\footnote{Note that this is completely distinct from the Robertson uncertainty relation \cite{Robertson},
$$
(\Delta A)_\rho (\Delta B)_\rho \geq \tfrac{1}{2}\left\vert \Tr{-i[A,B]\rho}\right\vert
$$
which pertains to the intrinsic quantum mechanical uncertainties of a pair of observables in the \emph{same} quantum state, $\rho$.} is a special case of the general story of successive measurement of non-commuting observables.

Consider measuring an observable $A$, with detector resolution $\sigma_A$, followed by a measurement of $B$, with detector resolution $\sigma_B$. We have
\[
\begin{split}
(\Delta A)^2_{\text{meas}} &= (\Delta A)^2_\rho + \sigma_A^2\\
(\Delta B)^2_{\text{meas}} &= (\Delta B)^2_{\rho'} + \sigma_B^2
\end{split}
\]
where $\rho'$ is the post-measurement state \eqref{vonN2} produced by measuring $A$.

\begin{equation}\label{hard}
\begin{split}
(\Delta B)^2_{\rho'}  &= \Tr{B^2\rho'} - (\Tr{B \rho'})^2\\
&=\int dl\, \Tr{f(l-A)^\dagger B^2 f(l-A) \rho} - \left(\int dl\, \Tr{f(l-A)^\dagger B f(l-A) \rho}\right)^2
\end{split}
\end{equation}
In general, it's difficult to make any headway in simplifying \eqref{hard}. If $[A,[A,B]]=0$, then we can simplify things greatly:
\begin{equation}\label{better}
\begin{split}
(\Delta B)^2_{\rho'} = (\Delta B)^2_{\rho} &+ s\bigl(\Tr{\{-i[A,B],B\}\rho} -2 \Tr{B\rho}\Tr{-i[A,B]\rho}\bigr)\\
&- s^2 \bigl(\Tr{-i[A,B]\rho}\bigr)^2 + \int |f'|^2\, dl\, \Tr{(-i[A,B])^2\rho}
\end{split}
\end{equation}
where
\begin{equation}
s = i\int \overline{f}f'\,dl
\end{equation}
Using Cauchy-Schwarz, we can show
\begin{equation}\label{Cauchy}
\int |f'|^2\,dl -\left\vert\int\overline{f}f'\,dl\right\vert^2 \geq \frac{1}{4\sigma_A^2}
\end{equation}
and hence
\begin{equation}\label{stillComplicated}
\begin{split}
(\Delta B)^2_{\rho'} \geq (\Delta B)^2_{\rho} &+ s\left(\Tr{\{-i[A,B],B\}\rho}-2 \Tr{B\rho}\Tr{-i[A,B]\rho}\right)\\
&+ s^2 \left(\Tr{(-i[A,B])^2\rho}-\left(\Tr{-i[A,B]\rho}\right)^2 \right)\\
&+\frac{1}{4\sigma_A^2}\Tr{(-i[A,B])^2\rho}
\end{split}
\end{equation}
This is still fairly formidable, except in two special cases
\begin{itemize}
\item The commutator is a constant, $-i[A,B]=c$.
\item The function,  $f$, is real and hence $s=0$.
\end{itemize} 

\noindent
In both cases, the dependence of \eqref{stillComplicated} on $s$ drops out, either because $s=0$ or because the coefficient of $s^k$ vanishes.

In the first case, \eqref{better} reduces to
\begin{equation}
(\Delta B)^2_{\rho'} = (\Delta B)^2_{\rho} + c^2 \left(\int |f'|^2\,dl - s^2\right)
\end{equation}

\begin{equation}
(\Delta B)^2_{\rho'} \geq (\Delta B)^2_{\rho} +  \frac{c^2}{4\sigma_A^2}
\end{equation}
From this, the product of the measured uncertainties,
\begin{equation}\label{useRobertson}
\begin{split}
(\Delta B)^2_{\text{meas}}(\Delta A)^2_{\text{meas}}\geq& \left((\Delta B)^2_{\rho} + \sigma_B^2 +\frac{c^2}{4\sigma_A^2}\right)\left((\Delta A)^2_{\rho}+\sigma_A^2\right)\\
\geq & \frac{c^2}{2} + \sigma_A^2 \sigma_B^2  + \frac{c^2 \sigma_A^2}{4(\Delta A)_\rho^2} +\left(\frac{c^2}{4\sigma_A^2} +\sigma_B^2\right)(\Delta A)_\rho^2
\end{split}
\end{equation}
where we used the Robertson uncertainty relation, $(\Delta B)^2_{\rho}\geq \frac{c^2}{4(\Delta A)^2_{\rho}}$. The RHS is minimized for $(\Delta A)^2_{\rho} = \sqrt{\frac{c^2\sigma_A^2}{4} \left(\frac{c^2}{4\sigma_A^2}+\sigma_B^2\right)}$, so we finally obtain
\begin{equation}\label{withRobertson}
(\Delta B)^2_{\text{meas}}(\Delta A)^2_{\text{meas}}\geq \frac{1}{4} \left(|c|+\sqrt{c^2+4\sigma_A^2\sigma_B^2}\right)^2
\end{equation}
With $c=1$, this is our more-careful derivation \cite{Distler:2012eh} of Heisenberg's relation
\begin{equation}
(\Delta x)_{\text{meas}}(\Delta p)_{\text{meas}}\geq \frac{1}{2}\left(1+\sqrt{1+4\sigma_x^2\sigma_p^2}\right)
\end{equation}

More generally, instead of using the Robertson uncertainty relation in \eqref{useRobertson}, we could use the Schr\"odinger uncertainty relation \cite{Schroedinger},
\[
(\Delta A)^2_{\rho}(\Delta B)^2_{\rho}\geq \tfrac{1}{4} X
\]
where
\[
\begin{split}
X&= \left\vert \Tr{\{A,B\}\rho}-\Tr{A\rho}\Tr{B\rho}\right\vert^2 + \left\vert\Tr{-i[A,B]\rho}\right\vert^2\\
&= \left\vert \Tr{\{A,B\}\rho}-\Tr{A\rho}\Tr{B\rho}\right\vert^2 + c^2
\end{split}
\]
Instead of \eqref{withRobertson}, we obtain
\begin{equation}
(\Delta B)^2_{\text{meas}}(\Delta A)^2_{\text{meas}}\geq \frac{1}{4} \left(\sqrt{X}+\sqrt{c^2+4\sigma_A^2\sigma_B^2}\right)^2
\end{equation}

Returning to the case where $[[A,B],A]=0$ ( but $[A,B]$ not necessarily constant) and $f$ is real, we see that the same steps yield the inequality
\begin{equation}
(\Delta B)^2_{\text{meas}}(\Delta A)^2_{\text{meas}}\geq \frac{1}{4} \left(\sqrt{X}+\sqrt{c^2+4\sigma_A^2\sigma_B^2}\right)^2
\end{equation}
where here
\[
\begin{split}
X&= \left\vert \Tr{\{A,B\}\rho}-\Tr{A\rho}\Tr{B\rho}\right\vert^2 + \left\vert\Tr{-i[A,B]\rho}\right\vert^2\\
c^2&\equiv \Tr{(-i[A,B])^2\rho}\quad.
\end{split}
\]

\section*{Acknowledgements}
We would like to thank D.~Freed and S.~Weinberg for discussions. This material is based upon work supported by the National Science Foundation under Grant Numbers PHY--1521186 and PHY--1620610.

\section*{Appendices}
\appendix
\addcontentsline{toc}{section}{Appendices}

\section{Proof of Equation \eqref{Cauchy}}

\begin{ulemma}
For any $f(l)$ as in \eqref{fdist}, we have
$$
\int |f'|^2\,dl -\left\vert\int \overline{f}f'\, dl\right\vert^2\geq \frac{1}{4\sigma_A^2}
$$
\end{ulemma}
\begin{proof}
Consider, for real $\alpha,\beta,\gamma$ (and $\alpha^2+\beta^2>0$),
\[
\begin{split}
0&\leq \frac{1}{4\sigma_A^2} \left\Vert\alpha l f + i\beta \sigma_A f +\gamma\sigma_A^2 f'\right\Vert^2\\
&= \alpha^2+\beta^2 +\gamma^2\sigma_A^2\int |f'|^2\, dl - \alpha\gamma - 2 \beta\gamma\sigma_A s
\end{split}
\]
where we denote $\Vert g\Vert^2 = \int |g|^2\, dl$ and $s=i \int\overline{f}f'\, dl$.
Viewing the RHS as a quadratic in $\gamma$, its discriminant must be negative
$$
0\geq (\alpha+2\beta\sigma_A s)^2 -4(\alpha^2+\beta^2)\sigma_A^2 \int |f'|^2\,dl
$$
or
$$
\int |f'|^2\,dl \geq \frac{(\alpha+2\beta\sigma_A s)^2}{4(\alpha^2+\beta^2)\sigma_A^2}
$$
Subtracting $s^2$ from both sides,
$$
\int |f'|^2\,dl -s^2\geq \frac{\alpha^2(1-4\sigma_A^2s^2)+4\alpha\beta\sigma_A s}{4(\alpha^2+\beta^2)\sigma_A^2}
$$
The RHS is maximized, as a function of $\beta$, for $\beta=2\alpha\sigma_A s$. At that point, it is equal to $\frac{1}{4\sigma_A^2}$.
\end{proof}

\section{Proof of Equation \eqref{generaleq}}\label{proveLindblad}

We start with
\[
\rho(t) = \int e^{-i Ht} f(l-\widetilde{A}(t)) \rho(0) f(l-\widetilde{A}(t))^\dagger e^{i Ht}\, dl
\]
where
\[
f(l) = \frac{1}{(2\pi\sigma_A^2)^{\mathrlap{1/4}}}\quad e^{-l^2/4\sigma_A^2}
\]
Differentiating and assuming $[ \widetilde{A}(t),\dot{\tilde{A}} (t)]=0$, we have
\[
\begin{split}
\dot{\rho}&=-i[H,\rho]+\frac{1/2\sigma_A^2}{(2\pi\sigma_A^2)^{\mathrlap{1/2}}}\quad\int e^{-iHt}e^{-(l-\widetilde{A})^2/4\sigma_A^2}\left\{
(l-\widetilde{A})\dot{\tilde{A}} ,\rho(0) 
\right\}e^{-(l-\widetilde{A})^2/4\sigma_A^2}e^{iHt}dl\\
&=-i[H,\rho]+\frac{1/2\sigma_A^2}{(2\pi\sigma_A^2)^{\mathrlap{1/2}}}\quad
\begin{array}[t]{l}
\displaystyle{\int} e^{-iHt}\left\{
(l'-\widetilde{A})(\dot{\tilde{A}}),\right.\\
\left. e^{-(l-\widetilde{A})^2/4\sigma_A^2}e^{is(l-l')/2}
\rho(0)
e^{-(l-\widetilde{A})^2/4\sigma_A^2}e^{is(l-l')/2}
\right\}e^{iHt}dl dl' \frac{ds}{2\pi}
\end{array}
\end{split}
\]
Completing the square
$$
-(l-\widetilde{A})^2/4\sigma_A^2 +i s(l-l')/2 = -\frac{1}{4\sigma_A^2}(l-is\sigma_A^2-\widetilde{A})^2 - \tfrac{1}{4}s^2\sigma_A^2+ \tfrac{1}{2}i s (\widetilde{A}-l')
$$
and letting $l''=l-is\sigma_A^2$, the $l'$ integral can be done
\[
\begin{split}
\dot{\rho}=-i[H,\rho]+&\frac{1/2\sigma_A^2}{(2\pi\sigma_A^2)^{\mathrlap{1/2}}}\quad\int e^{-s^2\sigma_A^2/2}e^{-iHt} \left\{(i\delta'(s)-\delta(s)\widetilde{A})\dot{\tilde{A}},\right.\\
&\left.e^{is\widetilde{A}/2}e^{-(l''-\widetilde{A})^2/4\sigma_A^2} \rho(0)
e^{-(l''-\widetilde{A})^2/4\sigma_A^2} e^{is\widetilde{A}/2} 
\right\} e^{iHt} dl'' \frac{ds}{2\pi}
\end{split}
\]
Commuting the $e^{\pm iHt}$ through,
\[
\begin{split}
\dot{\rho}=-i[H,\rho]+&\frac{1/2\sigma_A^2}{(2\pi\sigma_A^2)^{\mathrlap{1/2}}}\quad\int e^{-s^2\sigma_A^2/2}
 \left\{(i\delta'(s)-\delta(s)X)Y,\right.\\
&\left.e^{isX/2}
e^{-iHt}e^{-(l''-\widetilde{A})^2/4\sigma_A^2} \rho(0)
e^{-(l''-\widetilde{A})^2/4\sigma_A^2}e^{iHt}e^{isX/2}\right\} 
\, dl'' \frac{ds}{2\pi}
\end{split}
\]
where 
\[
\begin{split}
  X &= e^{-i Ht} \widetilde{A} e^{i H t}\\
  Y &= e^{-i Ht} \dot{\tilde{A} }e^{i H t}\\ 
\end{split}
\]
Now we can do the Gaussian integral over $l''$ to turn $\rho(0)$ into $\rho(t)$
\[
\begin{split}
\dot{\rho}&=-i[H,\rho]+\frac{1}{2\sigma_A^2}\int e^{-s^2\sigma_A^2/2} \left\{(i\delta'(s)-\delta(s)X)Y, e^{isX/2}\rho(t)e^{isX/2}\right\}ds\\
&=-i[H,\rho]+ \frac{1}{2\sigma_A^2}\left(\left\{Y,\tfrac{1}{2}(X\rho+\rho X)\right\}-\left\{XY,\rho\right\}\right)\\
&=-i[H,\rho]+ \frac{1}{4\sigma_A^2}\left(X\rho Y + Y \rho X-\{XY,\rho\}\right)\qquad\qed
\end{split}
\]
\vfill\eject
\addcontentsline{toc}{section}{References}
\bibliographystyle{utphys}
\bibliography{measurement}

\providecommand{\href}[2]{#2}\begingroup\raggedright\begin{thebibliography}{10}

\bibitem{Distler:2012eh}
J.~Distler and S.~Paban, ``Uncertainties in successive measurements,''
  \href{http://dx.doi.org/10.1103/PhysRevA.87.062112}{{\em Phys. Rev.}
  {\bfseries A87} no.~6, (2013) 062112},
\href{http://arxiv.org/abs/1211.4169}{{\ttfamily arXiv:1211.4169 [quant-ph]}}.

\bibitem{Weinberg:2016axv}
S.~Weinberg, ``What happens in a measurement?,''
  \href{http://dx.doi.org/10.1103/PhysRevA.93.032124}{{\em Phys. Rev.}
  {\bfseries A93} (2016) 032124},
\href{http://arxiv.org/abs/1603.06008}{{\ttfamily arXiv:1603.06008
  [quant-ph]}}.

\bibitem{Lindblad:1975ef}
G.~Lindblad, ``On the generators of quantum dynamical semigroups,''
\href{http://dx.doi.org/10.1007/BF01608499}{{\em Commun. Math. Phys.}
  {\bfseries 48} (1976) 119}.

\bibitem{Gorini:1975nb}
V.~Gorini, A.~Kossakowski, and E.~C.~G. Sudarshan, ``Completely positive
  dynamical semigroups of {N} level systems,''
\href{http://dx.doi.org/10.1063/1.522979}{{\em J. Math. Phys.} {\bfseries 17}
  (1976) 821}.

\bibitem{Bassi:2003gd}
A.~Bassi and G.~C. Ghirardi, ``Dynamical reduction models,''
  \href{http://dx.doi.org/10.1016/S0370-1573(03)00103-0}{{\em Phys. Rept.}
  {\bfseries 379} (2003) 257},
\href{http://arxiv.org/abs/quant-ph/0302164}{{\ttfamily arXiv:quant-ph/0302164
  [quant-ph]}}.

\bibitem{vonNeumann}
J.~von Neumann, {\em Mathematical Foundations of Quantum Mechanics}.
\newblock Princeton University Press, Princeton, NJ, 1955.

\bibitem{peres}
A.~Peres, {\em Quantum Theory: Concepts and Methods}.
\newblock Kluwer, Dordrecht, Netherlands, 1995.

\bibitem{Mackey}
G.~W. Mackey, {\em Mathematical Foundations of Quantum Mechanics}.
\newblock W. A. Benjamin, New York, NY, 1963.

\bibitem{Heisenberg}
W.~Heisenberg, ``{\"U}ber den anschaulichen inhalt der quantentheoretischen
  kinematik und mechanik,'' {\em Z. Phys.} {\bfseries 43} (1927) 172.

\bibitem{Robertson}
H.~Robertson, ``The uncertainty principle,'' {\em Physical Review} {\bfseries
  34} (1929) 163.

\bibitem{Schroedinger}
E.~Schr{\"o}dinger, ``Zum {H}eisenbergschen unsch{\"a}rfeprinzip,'' {\em
  Sitzungsberichte der PreussischenAkademie der Wissenschaften,
  Physikalischmathematische Klasse} {\bfseries 14} (1930) 296{--}303.

\end{thebibliography}\endgroup

\end{document}